\documentclass[letterpaper, 10 pt, conference]{ieeeconf}  

\IEEEoverridecommandlockouts                              

\usepackage{float} 

\usepackage[normalem]{ulem}
\usepackage{graphicx}
\usepackage{booktabs}

\usepackage{subcaption}
\usepackage{cancel}
\usepackage[dvipsnames]{xcolor}
\usepackage{epsfig,epsf,epstopdf}
\usepackage[cmex10]{amsmath}
\usepackage{amsfonts}
\usepackage{amssymb}
\usepackage{cite}
\usepackage{paralist}
\usepackage{array}

\usepackage{mathrsfs}
\usepackage{multirow}
\usepackage{caption}

\usepackage{nomencl}

\newtheorem{theorem}{Theorem}
\newtheorem{problem}{Problem}
\newtheorem{assumption}{Assumption}
\newtheorem{definition}{Definition}
\newtheorem{remark}{Remark}
\newtheorem{lemma}{Lemma}

\usepackage[colorlinks=true,linkcolor=blue,citecolor=blue,urlcolor=blue]{hyperref}

\makenomenclature
\begin{document}
\title{Event-Driven Safe and Resilient Control of Automated and Human-Driven
Vehicles under EU-FDI Attacks}

\author{
Yi Zhang, Yichao Wang, Wei Xiao, Mohamadamin~Rajabinezhad, and Shan Zuo%
\thanks{Y. Zhang, Y. Wang, M.~Rajabinezhad, and S. Zuo are with the Department of Electrical and Computer Engineering, University of Connecticut, Storrs, CT 06269, USA (emails: yi.2.zhang@uconn.edu; yichao.wang@uconn.edu; mohamadamin.rajabinezhad@uconn.edu; shan.zuo@uconn.edu). W. Xiao is with the Department of Robotics Engineering, Worcester Polytechnic Institute, Worcester, MA 01609, USA (email: wxiao3@wpi.edu).}%
}

\maketitle
\thispagestyle{empty}
\pagestyle{empty} 

\begin{abstract}
This paper studies the safe and resilient control of Connected and Automated Vehicles (CAVs) operating in mixed traffic environments where they must interact with Human-Driven Vehicles (HDVs) under uncertain dynamics and exponentially unbounded false data injection (EU-FDI) attacks. These attacks pose serious threats to safety-critical applications. While resilient control strategies can mitigate adversarial effects, they often overlook collision avoidance requirements. Conversely, safety-critical approaches tend to assume nominal operating conditions and lack resilience to adversarial inputs. To address these challenges, we propose an event-driven safe and resilient (EDSR) control framework that integrates event-driven Control Barrier Functions (CBFs) and Control Lyapunov Functions (CLFs) with adaptive attack-resilient control. The framework further incorporates data-driven estimation of HDV behaviors to ensure safety and resilience against EU-FDI attacks. Specifically, we focus on the lane-changing maneuver of CAVs in the presence of unpredictable HDVs and EU-FDI attacks on acceleration inputs. The event-driven approach reduces computational load while maintaining real-time safety guarantees. Simulation results, including comparisons with conventional safety-critical control methods that lack resilience, validate the effectiveness and robustness of the proposed EDSR framework in achieving collision-free maneuvers, stable velocity regulation, and resilient operation under adversarial conditions.
\end{abstract}

\section{INTRODUCTION}
In the context of CAVs, false data injection (FDI) attacks pose a growing threat to safety-critical applications \cite{weng2023secure}. These attacks can significantly disrupt cooperative control mechanisms by corrupting sensory or communicated data, ultimately leading to unsafe driving behaviors or even collisions. As vehicle systems increasingly rely on communication and automation, the vulnerability to such adversarial interventions becomes more pronounced. Traditional mitigation strategies often focus on detecting and removing compromised agents to eliminate the influence of malicious data \cite{zhang2021co,pasqualetti2013attack}. However, these approaches generally depend on restrictive assumptions, such as knowing an upper bound on the number of affected entities. In realistic and dynamic driving environments, especially in mixed traffic, such assumptions may be too strong and impractical to enforce.

To overcome these challenges, recent research has shifted toward developing resilient control strategies that maintain performance without explicitly identifying compromised components \cite{zuo2023resilient,zuo2022adaptive,zhang2024resilient,sun2023resilience,zuo2022resilient,rajabinezhad2025distributed,zuo2024data,zhang2025privacy,zhang2026observer,wang2025cyber,zhang2024distributed}. These methods aim to minimize the impact of adversarial actions by enhancing the intrinsic resilience of the control protocols. Nonetheless, many of these efforts overlook safety-critical constraints, such as collision avoidance, which are paramount in dense or high-speed traffic scenarios. Conversely, control strategies that emphasize safety through constrained optimization—such as using CBFs—often lack resilience to adversarial influence and may assume benign operational conditions \cite{garg2022fixed,cohen2024constructive,wu2023quadratic,cohen2023characterizing,tan2024undesired}.

This paper addresses the challenge of designing control strategies for CAVs that jointly ensure safety and resilience, particularly in the presence of EU-FDI attacks. Our focus is on guaranteeing collision avoidance even when the system is subjected to adversarial inputs. To this end, we investigate safe interactions between CAVs and HDVs in a mixed-traffic environment, where HDV dynamics and driver intentions are uncertain, nonlinear, and generally unknown.

To address this challenge, we revisit the event-driven CBF framework in \cite{xiao2022event} and redesign it to remain valid under EU-FDI attacks. Existing event-driven CBF methods rely on the premise that the dynamics used in barrier construction sufficiently match the true system evolution under bounded uncertainties, which guarantees meaningful Lie-derivative constraints and maintains quadratic program (QP) feasibility. This premise fundamentally breaks down under EU-FDI attacks. Exponentially unbounded adversarial inputs injected into the acceleration channel induce a rapidly growing mismatch between the modeled and actual longitudinal dynamics, which directly distorts the evolution of safety constraints. As a result, the CBF conditions may no longer certify forward invariance of the safe set, and the event-driven QP can become infeasible even though it remains valid under nominal or bounded disturbances. This distinction highlights a fundamentally different research problem from prior resilience studies. While existing EU-FDI works mainly investigate stability or consensus degradation under adversarial inputs with the control structure remaining intact, the problem addressed here concerns the loss of safety feasibility itself, where adversarial injections invalidate the safety-constrained optimization formulation. Therefore, this work does not combine existing methods but instead resolves a previously unaddressed failure mode of event-driven safety-critical control under EU-FDI attacks. The main contributions of this paper are summarized as follows:

\begin{itemize}
    \item To the best of our knowledge, this paper is the first to reveal and resolve the loss of validity of event-driven CBF-based safety-critical control under EU-FDI attacks by redesigning the safety mechanism to operate under adversarial injections. This unique combination ensures safety and resilience for CAVs interacting with HDVs in uncertain mixed traffic environments.
    
    \item The proposed EDSR framework effectively addresses EU-FDI attacks, providing provable collision avoidance and stable operation guarantees. Compared to the conventional safety-critical control method in \cite{xiao2022event}, simulation results demonstrate that alternative approaches fail to mitigate adversarial impacts, while the EDSR framework's robustness and effectiveness are clearly validated under these conditions.
\end{itemize}

\section{PROBLEM FORMULATION}
Figure~\ref{fig:1} illustrates a mixed-traffic lane-changing scenario.
Vehicles $A$ and $B$ are cooperative CAVs, $H$ is an uncontrollable HDV,
and $U$ is a slower vehicle ahead.
Vehicle $B$ initiates a lane change to overtake $U$. We focus on a high-speed lane-changing interaction between CAVs and HDVs,
which represents a worst-case conflict topology where longitudinal safety,
lateral maneuvering, and human unpredictability coexist.
The framework is general and extends to other conflict scenarios. The objectives are to minimize maneuver time and control effort,
limit traffic disruption, ensure safety under unpredictable HDV behavior,
and guarantee stability and resilience under malicious FDI attacks.
 
\begin{figure}[H]
\includegraphics[width =0.5 \textwidth]{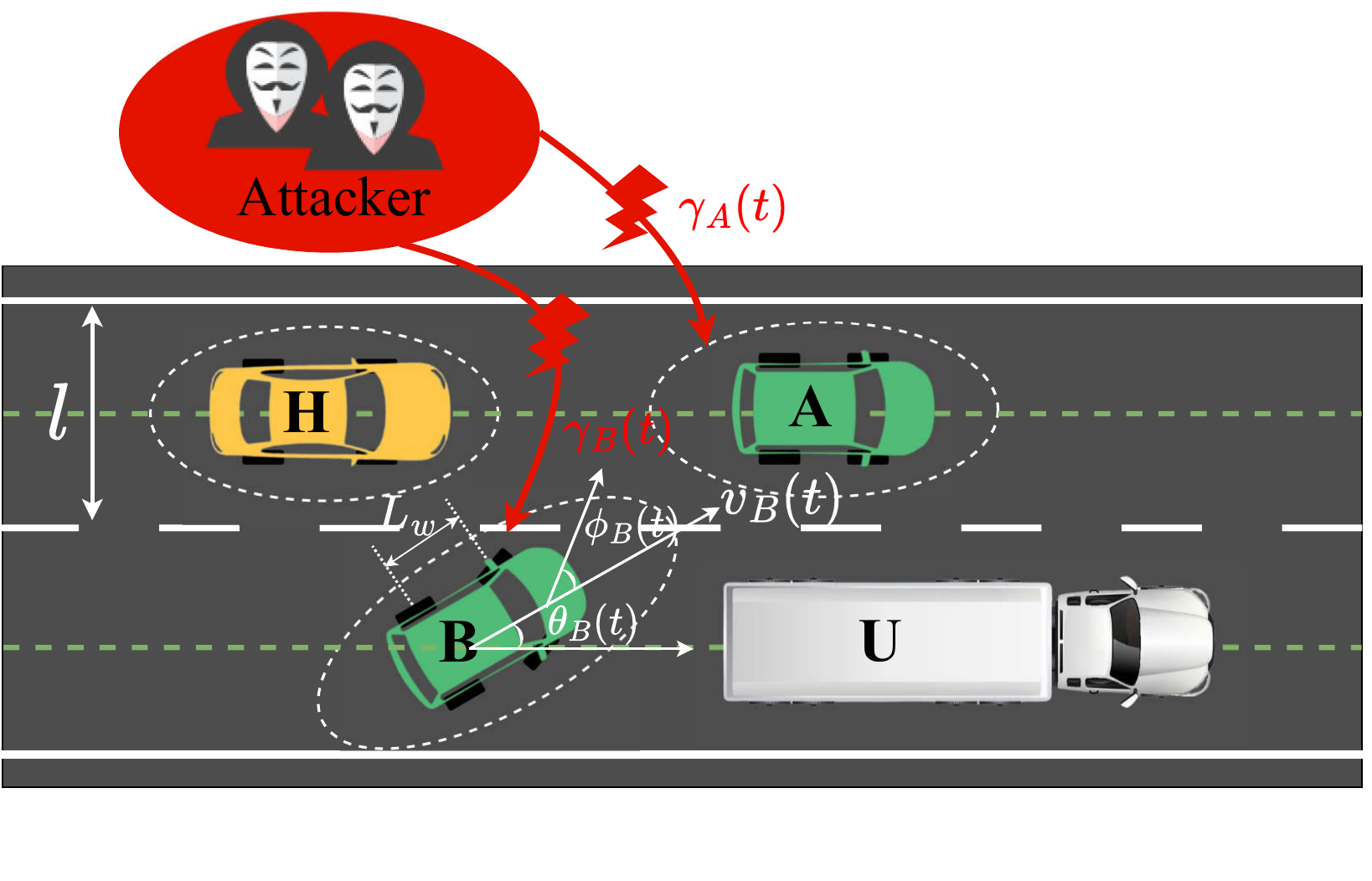}
\caption{Illustration of a basic lane-changing maneuver: the yellow vehicle is an HDV, green vehicles are CAVs, and the grey vehicle is a slow-moving, uncontrollable vehicle.}
\label{fig:1}
\end{figure}

In this scenario, the dynamics and control policy of the HDV are assumed to be unknown. The slow vehicle $U$ is assumed to continue moving in the slow lane at a constant speed $v_{U}$. For each CAV, indexed by $i \in \{A, B\}$ in Figure~\ref{fig:1}, the dynamics are defined according to the model in \cite{he2021rule} as follows:
\begin{equation}
\label{eq1}
\underbrace{
\begin{bmatrix}
\dot{x}_i(t)  \\
\dot{y}_i(t) \\
\dot{\theta}_i(t) \\
\dot{v}_i(t)    
\end{bmatrix}
}_{\dot{\boldsymbol{x}}_i(t)}=\underbrace{\begin{bmatrix}v_i(t) \cos \theta_i(t) \\
v_i(t) \sin \theta_i(t) \\
0 \\
0\end{bmatrix}
}_{f\left(\boldsymbol{x}_i(t)\right)}+\underbrace{\begin{bmatrix}0 & -v_i(t) \sin \theta_i(t) \\
0 & v_i(t) \cos \theta_i(t) \\
0 & v_i(t) / L_{w} \\
1 & 0\end{bmatrix}}_{g\left(\boldsymbol{x}_i(t)\right)} \underbrace{\begin{bmatrix}u_i(t) \\
\phi_i(t)\end{bmatrix}}_{\boldsymbol{\xi}_i(t)},
\end{equation}
Here, $x_i(t)$, $y_i(t)$, $\theta_i(t)$, and $v_i(t)$ denote the longitudinal position, lateral position, heading angle, and speed, respectively. The control inputs $u_i(t)$ and $\phi_i(t)$ represent the acceleration and steering angle.  
Here, $L_w$ denotes the wheelbase length of the vehicle, which affects the rate of change of the heading angle $\theta_i(t)$ during steering maneuvers.

Consider the following malicious FDI attacks targeting the acceleration control input of each CAV:
\begin{equation}
\label{eq5}
    \dot{v}_i(t) = \bar{u}_i(t) = u_i(t)+\gamma_i(t), \;i\in\{A,B\},
\end{equation}
where \( \bar u_i(t) \) is the corrupted control input. $\gamma_i(t)$ is the EU-FDI attack, which satisfies the following assumption.

\begin{assumption}
\label{ass:1}
$\gamma_i(t)$ is any FDI attacks with the worst case of being an exponentially growing signal and takes the general form of $\eta_i(t) \exp(\kappa_i(t) t)$. We assume $\left|\eta_i(t)\right|\leq \bar{\eta}$ and $\left|\kappa_i(t)\right| \leq \bar{\kappa}$, where $\bar{\eta}$ and $\bar{\kappa}$ are positive constants.  
\end{assumption}

The exponential form captures the fastest transient deviation that an adversarial input can induce within a short time window, representing a worst-case actuator corruption scenario. The maneuver starts at $t_0$ and ends at $t_f$, when vehicle $B$ completes the lane change.
During this period, the CAV inputs and velocity are constrained by
\begin{equation}
\label{eq2}
\begin{aligned}
u_{i,\min} \le u_i(t) \le u_{i,\max}, \quad
\phi_{i,\min} \le \phi_i(t) \le \phi_{i,\max}, \\
v_{i,\min} \le v_i(t) \le v_{i,\max}, \quad i \in \{A,B\}.
\end{aligned}
\end{equation}

Let $l$ denote the lane width and set the slow-lane centerline as $y=0$.
The lateral position satisfies
\begin{equation}
\label{eq3}
-\frac{l}{2} \le y_i(t) \le \frac{3}{2}l, \quad i \in \{A,B\}.
\end{equation} 

To ensure inter-vehicle safety, an ellipsoidal safe region is defined as
\begin{equation}
\label{eq4}
\begin{aligned}
b_{i,j}(\boldsymbol{x}_i,\boldsymbol{x}_j)
= \frac{(x_j-x_i)^2}{(a_i v_i)^2}
+ \frac{(y_j-y_i)^2}{(b_i v_i)^2}
-1 \ge 0,
\end{aligned}
\end{equation}
where $a_i$ and $b_i$ denote the longitudinal and lateral safety margins, respectively, with $i \in \{A,B\}$ and $j \in \{A,B,U,H\}$, $i \neq j$. The ellipsoidal form is differentiable and facilitates CBF construction, and \eqref{eq2} ensures $v_i(t)\geq v_{i,\min}>0$, so that \eqref{eq4} is well-defined.
\begin{definition}[\cite{khalil2002nonlinear}]
\label{def:1}
$x(t) \in \mathbb{R}$ is uniformly ultimately bounded (UUB) with the ultimate bound $b$, if there exist constants $b, c>0$, independent of $t_{0} \geq 0$, and for every $a \in(0, c)$, there exists $t_{1}=t_{1}(a, b) \geq 0$, independent of $t_{0}$, such that if $\left|x\left(t_{0}\right)\right| \leq a $, then $|x(t)| \leq b, \quad \forall t \geq t_{0}+t_{1}.$    
\end{definition}

Now, we introduce the safe and attack-resilient (SAR) control problem.
\begin{problem}
\label{prb:1}
Design control inputs $u_i(t)$ such that: 1) The lane-change maneuver minimizes control effort and terminal deviation;
2) Safety constraints $b_{i,j}(\cdot) \ge 0$ are satisfied for all $t$;
3) The speed regulation error $\varepsilon_i(t)$ is UUB under EU-FDI. This yields the following optimal control problem:
\begin{flalign}
\label{eq8}
\begin{split}
&\min_{u_A(t),u_B(t),t_f}  \int_{t_0}^{t_f} \frac{\alpha_u}{2} \left(u_{A}^{2}(t) + u_{B}^2(t)\right) dt  + \alpha_l \left(y_{B}(t_f) - l\right)^2
\\&+ \alpha_v \left[ \left(v_{A}(t_f) - v_d\right)^2 + \left(v_{B}(t_f) - v_d\right)^2 \right], \\
&\quad\quad\quad\quad\quad\text{s.t.} \quad  \eqref{eq1},   \eqref{eq2}, \eqref{eq3}, \eqref{eq4},
\end{split}&&\raisetag{2.5\baselineskip}
\end{flalign}
where $v_d$ is the desired cruising speed in the fast lane, and $l$ is the centerline of the target lane. $\alpha_u$, $\alpha_l$, and $\alpha_v$ are non-negative weights that represent the trade-offs between energy consumption, lateral deviation from the target lane, and velocity tracking performance, respectively.
\end{problem}
\begin{remark}
\label{rem:3}
The proposed EDSR framework integrates event-driven CBFs, CLFs,
and adaptive attack compensation to ensure safety and resilience.
Unlike \cite{xiao2022event}, EU-FDI corrupts the safety constraint dynamics
and may destroy QP feasibility, which necessitates a restructuring
of the event-driven CBF mechanism.
The ellipsoidal safe region in~\eqref{eq4}
provides a speed-dependent safety representation,
while the event-driven implementation reduces computational burden.
\end{remark}

\section{HUMAN-IN-THE-LOOP SAFE LANE-CHANGING MANEUVER}
In this section, an attack-resilient event-driven safe control framework is proposed, which incorporates an attack-resilient controller augmented with event-driven CBFs and CLFs, to address the SAR control problem defined in Problem \ref{prb:1}. To enforce safety under adversarial conditions, the original state constraints in Problem \ref{prb:1} are replaced by CBF-based constraints. Specifically, for any state constraint of the form $b(x) \geq 0$, a corresponding CBF condition is introduced. \begin{lemma}[ \cite{xiao2022event}]
\label{lem:1}
Let $b(\boldsymbol{x}(t))$ be a continuously differentiable function representing a state constraint such that the safe set is defined as $\mathcal{C} = \{ \boldsymbol{x}(t) \in \mathbb{R}^n \mid b(\boldsymbol{x}(t)) \geq 0 \}$. Assume that $L_g b(\boldsymbol{x}(t)) \neq 0$ when $b(\boldsymbol{x}(t)) = 0$. Then the original constraint $b(\boldsymbol{x}(t)) \geq 0$ is implied by the following condition:
\begin{equation}
\label{eq9}
\begin{aligned}
\sup_{\boldsymbol{\xi}(t) \in \Xi}& \left[ L_f b(\boldsymbol{x}(t)) + L_g b(\boldsymbol{x}(t)) \boldsymbol{\xi}(t) + \alpha(b(\boldsymbol{x}(t))) \right] \geq 0,  
\\&\forall \boldsymbol{x}(t) \in \mathcal{C},
\end{aligned}
\end{equation}
where $L_f(\cdot)$ and $L_g(\cdot)$ denote the Lie derivatives of $b(\boldsymbol{x}(t))$ along the system dynamics $f(\cdot)$ and input matrix $g(\cdot)$, respectively, and $\alpha(\cdot)$ is an extended class $\mathcal{K}$ function. This condition is affine in the control input and ensures forward invariance of the safe set $\mathcal{C}$.
\end{lemma}
\begin{remark}
\label{rem:4}
Solving the SAR control problem \ref{prb:1} with event-driven CBFs requires accurately estimating the unknown dynamics of HDVs. With these HDV state estimates, the CAVs then enforce safety by transforming original constraints into corresponding CBF constraints at every control update. By doing so, the original SAR control problem can be recast into a series of QPs, similar in form to those discussed in \cite{ames2014control}. An event-driven method further enhances efficiency by adaptively determining the appropriate triggering moments for solving these QPs.
\end{remark}

\subsection{Event-Driven Safety Control}

To approximate the unknown behavior of HDVs, adaptive dynamics are introduced based on real-time measurements. The HDV is modeled using the following adaptive nonlinear system:
\begin{equation}
\label{eq10}
\dot{\overline{\boldsymbol{x}}}_{H}(t) = f_{a}(\overline{\boldsymbol{x}}_{H}(t)) + g_{a}(\overline{\boldsymbol{x}}_{H}(t)) \boldsymbol{u}_{H}(t),
\end{equation}
where $f_{a}(\cdot): \mathbb{R}^{n} \rightarrow \mathbb{R}$ and $g_{a}(\cdot): \mathbb{R}^{n} \rightarrow \mathbb{R}^{n \times q}$ are adaptive functions designed to approximate the unknown dynamics of the HDV, $\boldsymbol{u}_{H}(t) \in \mathbb{R}^{q}$ is the HDV control input, and $\overline{\boldsymbol{x}}_{H}(t) \in X \subset \mathbb{R}^{n}$ is the estimated HDV state vector corresponding to the true state $\boldsymbol{x}_{H}(t)$.

Let the estimation error be defined as $\boldsymbol{e}(t) = \boldsymbol{x}_{H}(t) - \overline{\boldsymbol{x}}_{H}(t)$. For any safety constraint $b_{i,H}(\boldsymbol{x}_{i}(t), \boldsymbol{x}_{H}(t))$ involving a CAV $i$ and the HDV, the function can be re-expressed using the estimated state and the error:
\begin{equation}
\label{eq11}
b_{i,H}(\boldsymbol{x}_{i}(t), \boldsymbol{x}_{H}(t)) = b_{i,H}(\boldsymbol{x}_{i}(t), \overline{\boldsymbol{x}}_{H}(t) + \boldsymbol{e}(t)) \geq 0.
\end{equation}

To enforce this constraint through a control barrier function (CBF), the following condition is constructed:
\begin{equation}
\label{eq12}
\begin{aligned}
& \frac{\partial b_{i,H}(\boldsymbol{x}_{i}(t), \overline{\boldsymbol{x}}_{H}(t) + \boldsymbol{e}(t))}{\partial \boldsymbol{x}_{i}} \dot{\boldsymbol{x}}_{i}(t) 
\\&+ \frac{\partial b_{i,H}(\boldsymbol{x}_{i}(t), \overline{\boldsymbol{x}}_{H}(t) + \boldsymbol{e}(t))}{\partial \overline{\boldsymbol{x}}_{H}(t)} \dot{\overline{\boldsymbol{x}}}_{H}(t) \\
& + \frac{\partial b_{i,H}(\boldsymbol{x}_{i}(t), \overline{\boldsymbol{x}}_{H}(t) + \boldsymbol{e}(t))}{\partial \boldsymbol{e}(t)} \dot{\boldsymbol{e}}(t) 
\\&+ \alpha\left(b_{i,H}(\boldsymbol{x}_{i}(t), \overline{\boldsymbol{x}}_{H}(t) + \boldsymbol{e}(t))\right) \geq 0,
\end{aligned}
\end{equation}
where $\dot{\boldsymbol{x}}_{i}(t)$ is governed by the system dynamics in \eqref{eq1}, and $\dot{\overline{\boldsymbol{x}}}_{H}(t)$ is defined in \eqref{eq10}. Since $\boldsymbol{e}(t)$ and $\dot{\boldsymbol{e}}(t)$ can be directly evaluated from real-time sensor measurements, all terms in \eqref{eq12} are computable. Thus, satisfaction of \eqref{eq12} ensures that the original safety constraint $b_{i,H}(\boldsymbol{x}_{i}(t), \boldsymbol{x}_{H}(t)) \geq 0$ holds, even though the true HDV dynamics are unknown, as established in \cite{xiao2022event}.

To adaptively refine the HDV dynamics model and reduce unnecessary conservativeness, the function $f_{a}(\cdot)$ in \eqref{eq10} is updated at discrete time instants $t_k$, for $k = 1, 2, \ldots$, according to:
\begin{equation}
\label{eq13}
f_{a}(\overline{\boldsymbol{x}}_{H}(t_k^+)) = f_{a}(\overline{\boldsymbol{x}}_{H}(t_k^-)) + \dot{\boldsymbol{e}}(t_k),
\end{equation}
where $t_k^-$ and $t_k^+$ represent the time instants immediately before and after $t_k$. By resetting the estimated state to the measured state at each update, i.e., $\overline{\boldsymbol{x}}_{H}(t_k) = \boldsymbol{x}_{H}(t_k)$, the error satisfies $\boldsymbol{e}(t_k) = \mathbf{0}$ and $\dot{\boldsymbol{e}}(t_k^+) \approx \mathbf{0}$. This approach improves control performance by reducing the frequency of event-driven QP updates and limiting unnecessary conservativeness in the controller.

In a lane-changing scenario, the dynamics of the HDV traveling in the fast lane are modeled adaptively to account for unknown behavior. The estimated HDV dynamics are defined as:
\begin{equation}
\label{eq14}
\underbrace{\begin{bmatrix}\dot{\bar{x}}_{H}(t) \\
\dot{\bar{y}}_{H}(t) \\
\dot{\bar{\theta}}_{H}(t) \\
\dot{\bar{v}}_{H}(t)\end{bmatrix}}_{\dot{\bar{\boldsymbol{x}}}_{H}(t)} = \underbrace{\begin{bmatrix}
\bar{v}_{H}(t) \cos \bar{\theta}_{H}(t) + h_{x}(t) \\
\bar{v}_{H}(t) \sin \bar{\theta}_{H}(t) + h_{y}(t) \\
\bar{v}_{H}(t) / L_{w} + h_{\theta}(t) \\
h_{v}(t)
\end{bmatrix}}_{f_{a}(\overline{\boldsymbol{x}}_{H}(t))},
\end{equation}
where $\overline{\boldsymbol{x}}_{H}(t) = [\bar{x}_{H}(t), \bar{y}_{H}(t), \bar{\theta}_{H}(t), \bar{v}_{H}(t)]^\top$ represents the estimated longitudinal and lateral positions, heading angle, and velocity of the HDV. The terms $h_{x}(t), h_{y}(t), h_{\theta}(t), h_{v}(t) \in \mathbb{R}$ are adaptive components used to approximate unknown variations in HDV behavior. These are initialized as zero at the start time: $h_{j}(t_0) = 0$ for $j \in \{x, y, \theta, v\}$.
 
To ensure that the constraints in the SAR control problem—namely, the speed and control constraints \eqref{eq2}, position constraints \eqref{eq3}, and safety constraints \eqref{eq4}—are satisfied at all times, CBFs are employed. Each of these constraints can be expressed as $b_n(\boldsymbol{x}(t)) \geq 0$ for $n \in \{1, 2, \ldots, N\}$, where $\boldsymbol{x}(t)$ includes all relevant vehicle states, i.e., $\boldsymbol{x}(t) = \{\boldsymbol{x}_{A}(t), \boldsymbol{x}_{B}(t), \boldsymbol{x}_{H}(t), \boldsymbol{x}_{U}(t)\}$. The corresponding CBF condition for each constraint takes the general form:
\begin{equation}
\label{eq15}
L_{f} b_n(\boldsymbol{x}(t)) + L_{g} b_n(\boldsymbol{x}(t)) \boldsymbol{\xi}(t) + \alpha_n(b_n(\boldsymbol{x}(t))) \geq 0,
\end{equation}
where $\boldsymbol{\xi}(t) = \{\boldsymbol{\xi}_{A}(t), \boldsymbol{\xi}_{B}(t)\}$.
  
In addition to CBFs that enforce hard safety constraints, CLFs are introduced to encode soft performance objectives related to lane-keeping and velocity tracking. Based on the terminal cost in \eqref{eq8}, the following CLFs are defined:
\begin{align*}
&V_{1}(\boldsymbol{x}_{B}(t)) = (v_{B}(t) - v_d)^2,  
V_2(\boldsymbol{x}_{A}(t)) = (v_{A}(t) - v_d)^2, \\
&V_3(\boldsymbol{x}_{B}(t)) = (y_{B}(t) - l)^2, 
 V_4(\boldsymbol{x}_{A}(t)) = (y_{A}(t) - l)^2.
\end{align*}

The corresponding CLF constraints are imposed separately on the longitudinal and lateral control channels, yielding
\begin{equation}
\label{eq16}
\begin{aligned}
L_f V_j(\boldsymbol{x}) + L_g V_j(\boldsymbol{x}) u_i^s + c_3 V_j(\boldsymbol{x}) &\leq \delta_j, \\
L_f V_j(\boldsymbol{x}) + L_g V_j(\boldsymbol{x}) \phi_i + c_3 V_j(\boldsymbol{x}) &\leq \delta_j,
\end{aligned}
\end{equation}
where the first inequality corresponds to the longitudinal channel (acceleration control), and the second corresponds to the lateral channel (steering control), with $i\in\{A,B\}$ and $j\in\{1,2,3,4\}$. Here, $u_i^s(t)$ denotes the nominal safety control input and the variables $\delta_j$ are relaxation terms introduced to preserve QP feasibility.

\begin{lemma}[\cite{li2024safe}]
\label{lem:2}
Consider a system of CAVs interacting with HDVs whose dynamics are unknown but estimated online. Let the safety constraints be encoded by CBFs $b_n(\boldsymbol{x})$ with Lie derivatives $L_f b_n, L_g b_n$. Then, the QP enforcing safety can be solved at a sequence of event-driven times $\{t_k\}_{k=0}^\infty$, rather than fixed intervals, if the following conditions are satisfied:

\noindent 1) The HDV state estimation error $\boldsymbol{e}(t)$ and its derivative $\dot{\boldsymbol{e}}(t)$ are bounded as:
\begin{subequations}
\label{eq18}
\begin{flalign}
\label{eq18a}|e_x(t)| &\leq w_x,  |e_y(t)| \leq w_y,  |e_\theta(t)| \leq w_\theta,  |e_v(t)| \leq w_v,&&\raisetag{0.7\baselineskip}  \\\label{eq18b}
    |\dot{e}_x(t)| &\leq \nu_x,  |\dot{e}_y(t)| \leq \nu_y,  |\dot{e}_\theta(t)| \leq \nu_\theta,  |\dot{e}_v(t)| \leq \nu_v.&&\raisetag{0.7\baselineskip} 
    \end{flalign}
    \end{subequations}
\noindent2) For each CBF constraint $b_n(\boldsymbol{x})$, the following robust condition must hold:
    \begin{equation}
    \label{eq19}
    L_{f_{\min}} b_n(t_k) + L_{g_{\min}} b_n(t_k) \boldsymbol{\xi}(t) \geq 0,
    \end{equation}
    where $L_{f_{\min}} b_n(t_k) = \min_{\boldsymbol{r} \in R(t_k)} \left[ L_f b_n(\boldsymbol{z}) + \alpha_n(b_n(\boldsymbol{z})) \right]$,
    and $L_{g_{\min}} b_n(t_k)$ is determined based on the control direction signs.

\noindent3) Control inputs $u_i^s(t_k)$ are obtained by solving the QP:
\begin{equation}
    \label{eq20}
    \begin{aligned}\min_{u_i^s(t_k), \delta_j(t_k)}& \sum_{i=A,B} \alpha_{u_i^s} \left(u_i^s\right)^2(t_k) + \sum_{j=1}^{4} p_j \delta_j^2(t_k)
    \\\text{s.t.} \quad &\eqref{eq1},\eqref{eq16},\eqref{eq19},
    \end{aligned}
    \end{equation}
where $\alpha_{u_i^s}$ and $p_j$ are design parameters that balance control effort and constraint relaxation.

In \eqref{eq18a}-\eqref{eq20}, the event-driven time $t_{k+1}, k=0,1,2, \ldots$ is given by
\begin{flalign}
\label{eq21}
\begin{aligned}
&t_{k+1}=\min \{t>t_k:\;
\exists j\in\{x,y,\theta,v\}\ \text{s.t.}\ |e_j(t)|\ge w_j, \\
&\text{or } \exists j\in\{x,y,\theta,v\}\ \text{s.t.}\ |\dot e_j(t)|\ge \nu_j,
\\&\text{or } |\boldsymbol{x}_i(t)-\boldsymbol{x}_i(t_k)|\ge \boldsymbol{s}_i, \text{or } |\overline{\boldsymbol{x}}_H(t)-\overline{\boldsymbol{x}}_H(t_k)|\ge \boldsymbol{s}_H \},
\end{aligned}
&&\raisetag{1.5\baselineskip}
\end{flalign}
where $i \in\{A,B\}$. $\boldsymbol{s}_{i} \in \mathbb{R}_{\geq 0}^{4}$ is a error bound vector.
\end{lemma}

\begin{remark}
Although the triggering condition in \eqref{eq21} appears mathematically involved, 
all terms are directly measurable from onboard sensors, including the state estimation error, its derivative, and the state deviation from the last update instant. 
No additional optimization or model-based computation is required to evaluate the triggering rule. 
Therefore, the event-driven mechanism is computationally lightweight and fully implementable in real time. In implementation, time is discretized with a fixed sampling period $T_s$, and the control problem is solved at time instants $t_k = t_0 + kT_s$, for $k = 1,2,\ldots$. 
At each update instant, the SAR control problem \ref{prb:1} is reformulated as the QP \eqref{eq20}. This optimization ensures that the control input not only respects all safety constraints but also minimizes the deviation from the nominal control input. 
The solution is obtained by solving the QP in \cite{zhang2023novel}.
\end{remark}

\subsection{Attack-Resilient Control}

Define the following speed regulation error
\begin{equation}
\label{eq22}
\varepsilon_i(t) = v_i(t) - v_d,\quad i\in\{A,B\}.
\end{equation}

To address the EU-FDI attacks in \eqref{eq5}, the attack-resilient control framework is designed as follows:
\begin{align}
\label{eq23}
u_i(t)&= u_i^s(t) - \hat{\gamma}_i(t),
\\\label{eq24}
\hat{\gamma}_i(t) &=\frac{\varepsilon_i(t)}{\left|\varepsilon_i(t)\right|+\exp(-c_i t^2)} \exp(\hat{\rho}_i(t)),
\\\label{eq25}
\dot{\hat{\rho}}_i(t) &= \alpha_i  \left|\varepsilon_i(t)\right|,\quad i\in\{A,B\},
\end{align}
where compensating signal \( \hat{\gamma}_i(t) \) is designed to counteract the detrimental impact of the attack \( \gamma_i(t) \), while the parameter \( \hat{\rho}_i(t) \) is adaptively adjusted according to \eqref{eq25}, with a positive constant \( \alpha_i > 0 \).

\begin{remark}
\eqref{eq24} may appear unconventional at first glance, but it is designed to capture the physical intuition of adaptive compensation under adversarial conditions. $\varepsilon_i(t)$ represents the speed regulation error, which is a measurable physical quantity indicating deviation from the desired velocity. The denominator $|\varepsilon_i(t)| + \exp(-c_i t^2)$ ensures that the compensating signal remains bounded and smoothly transitions over time, mimicking a saturation-like behavior that prevents excessive control effort during transient phases. The exponential term $\exp(\hat{\rho}_i(t))$ reflects the accumulated history of the error, similar to an integral action in classical control. It amplifies the compensating signal when persistent deviations are observed, allowing the controller to respond more aggressively to sustained attacks. This structure is inspired by adaptive control principles, where the system learns to counteract disturbances based on observed performance. Physically, the controller behaves like a dynamic filter that estimates and cancels out the adversarial input $\gamma_i(t)$ using only observable quantities. Therefore, \eqref{eq23}--\eqref{eq25} together constitute an attack-resilient control framework that mitigates the adverse effects of EU-FDI attacks on the acceleration inputs by generating an adaptively tuned compensation signal.
\end{remark}

Unlike prior EU-FDI analyses where the control law structure remains intact and only stability margins are affected, here the entire safety QP formulation would become infeasible without the proposed attack-resilient compensation.

\begin{theorem}
\label{thm:1}
Consider the CAV dynamics in \eqref{eq1} and the adaptive HDV model in \eqref{eq10} under EU-FDI attacks. 
Assume that the initial state lies in the safe set, the estimation error satisfies \eqref{eq18a}--\eqref{eq18b}, and the event-driven QP \eqref{eq20} remains feasible at each triggering instant. Then, the objectives in Problem \ref{prb:1} are achieved implementing the event-driven safe and resilient controller in \eqref{eq18a}-\eqref{eq21} and \eqref{eq23}-\eqref{eq25}.
\end{theorem}

\begin{proof}
Based on Lemma~\ref{lem:1} and Lemma~\ref{lem:2}, the CAV achieves optimal trajectory tracking by minimizing control energy and terminal deviation from the desired lane position and speed, while the state $x_i(t)$ remains within a time-varying elliptical safe region for all $t \geq 0$. This safety is enforced through a QP that integrates CBFs and CLFs under an event-driven framework. Hence, the first two objectives of the SAR control problem are achieved. Next, we prove that the speed regulation error $\varepsilon_i(t)$ for each CAV is UUB. A Lyapunov function candidate for each CAV $i \in \{A,B\}$ is defined as $V(\varepsilon_i(t)) = (v_i(t) - v_d)^2$. Then, one has
\begin{equation}
\label{eq27}
\begin{aligned}
\dot{V}(\varepsilon_i(t)) =& 2(v_i(t) - v_d)\dot{v}_i(t)
\\=& 2 \varepsilon_i(t)  \bar{u}_i(t)
\\=& 2 \varepsilon_i (t) \left( u_i^s(t) - \hat{\gamma}_i(t) + \gamma_i(t) \right)
\end{aligned}
\end{equation}
For the velocity-tracking CLF defined as $V_i(\varepsilon_i)=\varepsilon_i^2(t)$, one has $L_f V_i = 0,$ and $L_g V_i = 2\varepsilon_i(t).$ Substituting them into the CLF condition \eqref{eq16} yields $2 \varepsilon_i(t) u_i^s(t) \leq -c_3 \varepsilon_i^2(t) + \delta_i(t).$ Then, choosing $\left|\varepsilon_i(t)\right|\geq\sqrt{\delta_i(t)/c_3}$, one has
\begin{flalign}
\label{eq29}
\begin{aligned}
&\dot{V}(\varepsilon_i(t)) 
= 2 \varepsilon_i(t)  u_i^s(t) - 2 \varepsilon_i(t)  \hat{\gamma}_i(t) + 2 \varepsilon_i(t)  \gamma_i(t)
\\\leq& -c_3 \varepsilon_i^2(t) + \delta_i(t) - \frac{2 |\varepsilon_i(t)|^2}{|\varepsilon_i(t)| + \exp(-c_i t^2)}  \exp(\hat{\rho}_i(t)) 
\\&+ 2 |\varepsilon_i(t)| |\gamma_i(t)|
\\\leq &2|\varepsilon_i(t)||\gamma_i(t)| - \frac{2|\varepsilon_i(t)|^2}{|\varepsilon_i(t)| + \exp(-c_i t^2)}  \exp(\hat{\rho}_i(t))
\\\leq&2 \left( |\varepsilon_i(t)|\right.
\\\times&\left.\left(\frac{|\varepsilon_i(t)||\gamma_i(t)| +\exp(-c_i t^2)|\gamma_i(t)| - |\varepsilon_i(t)| \exp(\hat{\rho}_i(t))}{|\varepsilon_i(t)| + \exp(-c_i t^2)} \right) \right)
\end{aligned}&&\raisetag{5\baselineskip}
\end{flalign}

By using the integral solution for differential equations in \eqref{eq25}, one has $\exp\left(\hat{\rho}_i(t)\right) = \exp\left(\hat{\rho}_i(0)\right) \exp\left(\alpha_i \int_0^\top \left| \varepsilon_i(\tau) \right| \mathrm{d} \tau\right)$. Define the compact set $\Upsilon_i \triangleq \{\left|\varepsilon_i(t)\right| \leq \kappa_i\}$. Pick $\hat{\rho}_i(0) = 0$ so that $\exp\left(\hat{\rho}_i(0)\right) = 1$. Then, outside the compact set $\Upsilon_i$, one has that there exists a time $t_1$ such that for all $t \geq t_1$, $|\varepsilon_i(t)| |\gamma_i(t)| - |\varepsilon_i(t)| \exp(\hat{\rho}_i(t)) \leq 0$. Considering Assumption \ref{ass:1}, one has $\lim_{t \to \infty} \exp(-c_i t^2) |\gamma_i(t)| = 0$. Hence, one obtains that outside the compact set $\Upsilon_i$, for all $t \geq t_1$, $\dot{V}(\varepsilon_i(t))  \leq 0$. Then, by LaSalle’s invariance principle \cite{lasalle1960some}, one has that $\varepsilon_i(t)$ is UUB. This completes the proof.
\end{proof}

\section{SIMULATION RESULTS} 

This section presents simulation results that demonstrate the effectiveness of the proposed framework in enabling each CAV to perform optimal lane-changing maneuvers with safety guarantees under EU-FDI attacks, even in the presence of HDVs whose dynamics are unknown to the CAVs. The selected lane-changing scenario is intentionally designed as a worst-case mixed-traffic interaction, where safety constraints, human uncertainty, and adversarial actuator corruption coexist. Unlike many existing works that assume cooperative or predictable surrounding vehicles, the considered scenario includes an HDV whose behavior is unknown and cannot be influenced by the controller, thereby introducing significant uncertainty into the traffic environment. Under EU-FDI attacks that corrupt the longitudinal dynamics of the CAVs, this unpredictability prevents reliance on structured interactions or cooperative maneuvers for safety recovery. Consequently, the scenario forms a highly adversarial test condition that highlights the necessity of the proposed EDSR framework in realistic mixed-traffic environments. The simulation setting corresponds to the mixed-traffic lane-changing scenario shown in Figure~\ref{fig:1} and Figure~\ref{fig:2}. Vehicle $U$ is assumed to travel at a constant speed of $v_{U}(t)=20~\mathrm{m/s}$ throughout the simulation, although this assumption is not required by the proposed control framework. The dynamics of vehicles $A$ and $B$ follow the model described in \eqref{eq1}. The allowable speed range for the CAVs is specified as $v_i(t)\in[15,35]~\mathrm{m/s}$, and their acceleration and steering inputs are constrained within $\boldsymbol{\xi}_i(t)\in[(-7,-\pi/4),(3.3,\pi/4)]~\mathrm{m/s}^2$. The desired speed $v_d$ is set to $30~\mathrm{m/s}$ to reflect the traffic flow velocity. To ensure safety, elliptical safe regions defined in \eqref{eq4} are employed, where the ellipse parameters are set as $a_A=a_B=0.6$ to represent the reaction time of the CAVs, and $b_A=b_B=0.1$ to approximate the lane width $l=4~\mathrm{m}$ via the minor axis. The maximum allowable time for completing the lane-changing maneuver is set to $T_f=15~\mathrm{s}$. The real dynamics of the HDV $H$, which are unknown to the controller, are given by \begin{equation*} \begin{aligned} \begin{bmatrix}\dot{x}_{H}(t) \\ \dot{y}_{H}(t) \\ \dot{\theta}_{H}(t) \\ \dot{v}_{H}(t)\end{bmatrix}=&\begin{bmatrix}v_{H}(t) \cos \theta_{H}(t) \\ v_{H}(t) \sin \theta_{H}(t) \\ 0 \\ 0\end{bmatrix}+\begin{bmatrix}0 & -v_{H}(t) \sin \theta_{H}(t) \\ 0 & v_{H}(t) \cos \theta_{H}(t) \\ 0 & v_{H}(t) / L_{w} \\ 1 & 0\end{bmatrix} \\&\times\begin{bmatrix}u_{H}(t) \\ \phi_{H}(t)\end{bmatrix} +\begin{bmatrix}\varepsilon_{1}(t) \\ \varepsilon_{2}(t) \\ \varepsilon_{3}(t) \\ \varepsilon_{4}(t)\end{bmatrix}, \end{aligned} \end{equation*} where $u_H(t)$ and $\phi_H(t)$ represent either a random policy or manual control by a human driver input, and the disturbances $\varepsilon_1(t)\in[-0.7,0.7]$, $\varepsilon_2(t)\in[-0.5,0.5]$, $\varepsilon_3(t)\in[-0.5,0.5]$, and $\varepsilon_4(t)\in[-0.7,0.7]$ are modeled as uniformly distributed random variables. The initial states of the vehicles at time $t_0=0$ are given as $\boldsymbol{x}_{A}(t_0)=[50~\mathrm{m},4~\mathrm{m},0~\mathrm{rad},29~\mathrm{m/s}]^\top$, $\boldsymbol{x}_{B}(t_0)=[20~\mathrm{m},0~\mathrm{m},0~\mathrm{rad},25~\mathrm{m/s}]^\top$, $\boldsymbol{x}_{H}(t_0)=[10~\mathrm{m},4~\mathrm{m},0~\mathrm{rad},28~\mathrm{m/s}]^\top$, and $\boldsymbol{x}_{U}(t_0)=[60~\mathrm{m},0~\mathrm{m},0~\mathrm{rad},20~\mathrm{m/s}]^\top$. The weights in the objective function \eqref{eq20} are selected as $\alpha_{u_A^s}=\alpha_{u_B^s}=1$, $p_1=p_2=p_4=1$, and $p_3=100$. The initial state of the adaptive HDV model \eqref{eq14} is set to match the true HDV state, i.e., $\overline{\boldsymbol{x}}_H(t_0)=\boldsymbol{x}_H(t_0)$, and the adaptive terms are initialized as $h_x(t_0)=h_y(t_0)=h_\theta(t_0)=h_v(t_0)=0$. The parameters $s_i$ in Lemma~\ref{lem:2} are set as $s_i=[0.01\,\mathrm{m},0.005\,\mathrm{m},0.01\,\mathrm{rad},1\,\mathrm{m/s}]^\top$ for all vehicles $i\in\{A,B,H,U\}$. The corresponding bounds for the state error $\boldsymbol{e}(t)$ and its derivative $\dot{\boldsymbol{e}}(t)$ are $\boldsymbol{w}=[0.2\,\mathrm{m},0.1\,\mathrm{m},0.1\,\mathrm{rad},1\,\mathrm{m/s}]^\top$ and $\boldsymbol{\nu}=[0.5\,\mathrm{m/s},0.2\,\mathrm{m/s},0.1\,\mathrm{rad/s},1\,\mathrm{m/s^2}]^\top$, respectively. Additionally, the allowable error threshold to terminate the maneuver is set as $\sigma=0.3\,\mathrm{m}$. 
\begin{figure}[H] \centering \begin{subfigure}[b]{0.2\textwidth} \includegraphics[width=\linewidth]{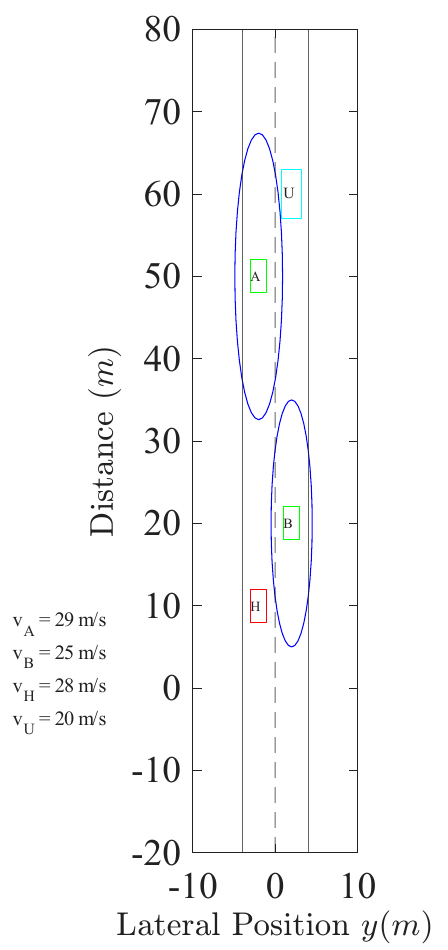} \caption{$t=t_0$.} \label{fig:2a} \end{subfigure} \begin{subfigure}[b]{0.2\textwidth} \includegraphics[width=\linewidth]{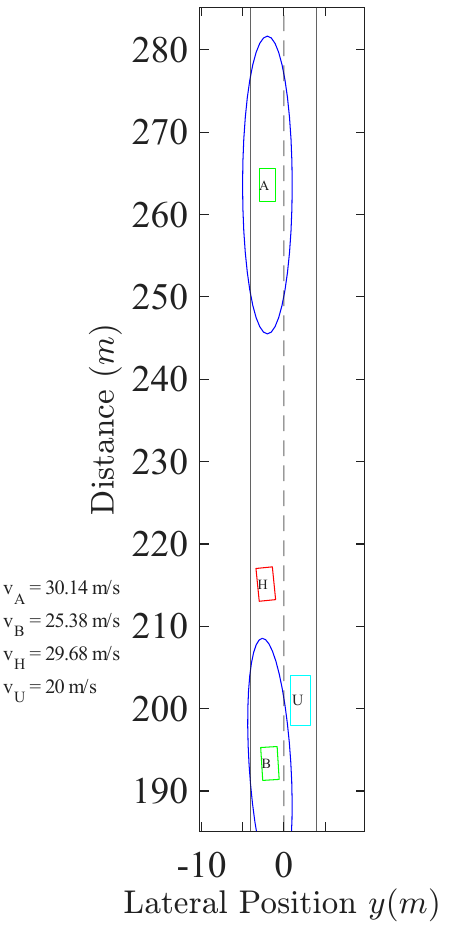} \caption{$t=t_f$.} \label{fig:2b} \end{subfigure} \caption{Lane-changing maneuver under EU-FDI attacks using the proposed EDSR framework.} \label{fig:2} \end{figure}

\subsection{Performance of the Proposed EDSR Framework under EU-FDI
}

Based on the above settings, we implement the proposed EDSR control framework in solving CBF-based QPs \eqref{eq20} with unknown HDV dynamics under EU-FDI attacks. Set the discretized time interval $T_s=0.05 \mathrm{~s}$. The HDV policy is set to be random, satisfying $u_{H}\left(t_{k}\right) \in[-1.7,1.7] \mathrm{m} / \mathrm{s}^{2}, \phi_{H}\left(t_{k}\right) \in$ $[-0.2 \pi, 0.2 \pi] \mathrm{rad},\quad k=0,1,2, \ldots$. As illustrated in Figure~\ref{fig:2}, the left configuration in Figure \ref{fig:2a} corresponds to the initial configuration at $t=t_0$, while the right configuration shows the final configuration in Figure \ref{fig:2b} at $t=t_f$ after CAV $B$ successfully completes the lane-changing maneuver under the proposed EDSR framework. These snapshots provide a qualitative overview of the maneuver evolution before presenting the quantitative performance results.

As shown in Figure~\ref{fig:3}, the EU-FDI attack signal grows unbounded over time, illustrating the exponential escalation of adversarial input injected into the system. The EU-FDI attacks are as follows
\begin{subequations}
\begin{align}
\label{eq30a}
\gamma_A(t)=2\exp(\kappa t)  \sin(5t)
\\\label{eq30b}\gamma_B(t)=5\exp(\kappa t)   \cos(5 t)
\end{align}
\end{subequations}

\begin{figure}[H]
\centering
\includegraphics[width =0.46 \textwidth]{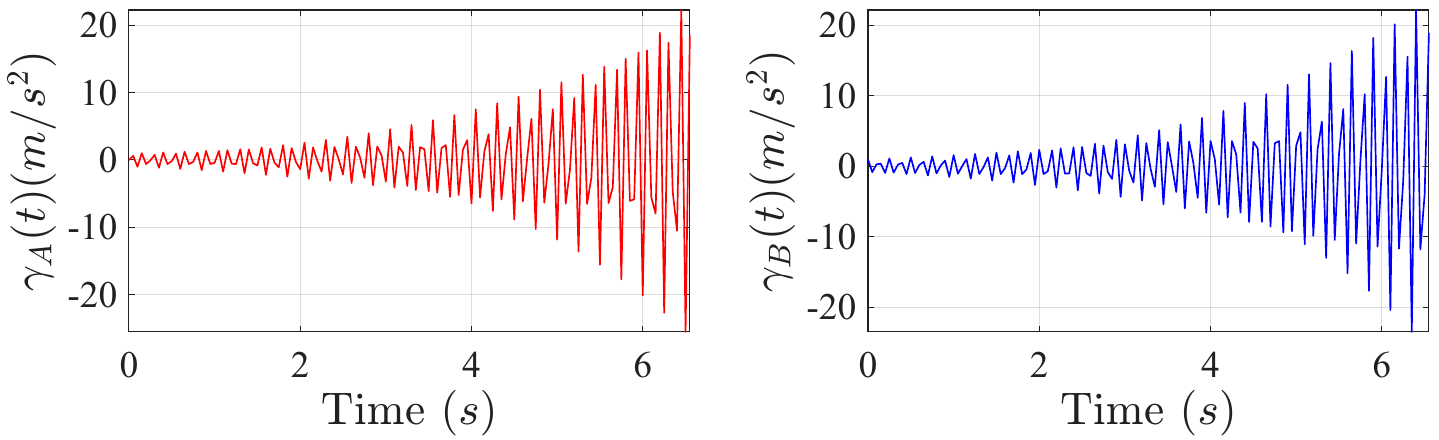}
\caption{EU-FDI Attack Signal.
}
\label{fig:3}
\end{figure}

The simulation results are shown in Figs.~\ref{fig:4}--\ref{fig:6}. 
Figure~\ref{fig:4} demonstrates that CAV $B$ successfully completes the lane-change maneuver within approximately $6.5$~s under EU-FDI attacks, while both CAVs maintain smooth velocity profiles without excessive acceleration or braking. Figure~\ref{fig:5} presents the safety constraint values $b_{i,j}$ for all interacting vehicle pairs. All curves remain above zero throughout the maneuver, confirming that the proposed EDSR framework preserves safety and maintains forward invariance of the safe set despite adversarial disturbances. Figure~\ref{fig:6} shows the speed tracking errors $\varepsilon_i(t)$ of CAVs $A$ and $B$. The errors remain smooth and bounded throughout the maneuver, indicating stable velocity regulation and effective mitigation of EU-FDI attacks.

\begin{figure}[H]
\centering
\includegraphics[width =0.46\textwidth]{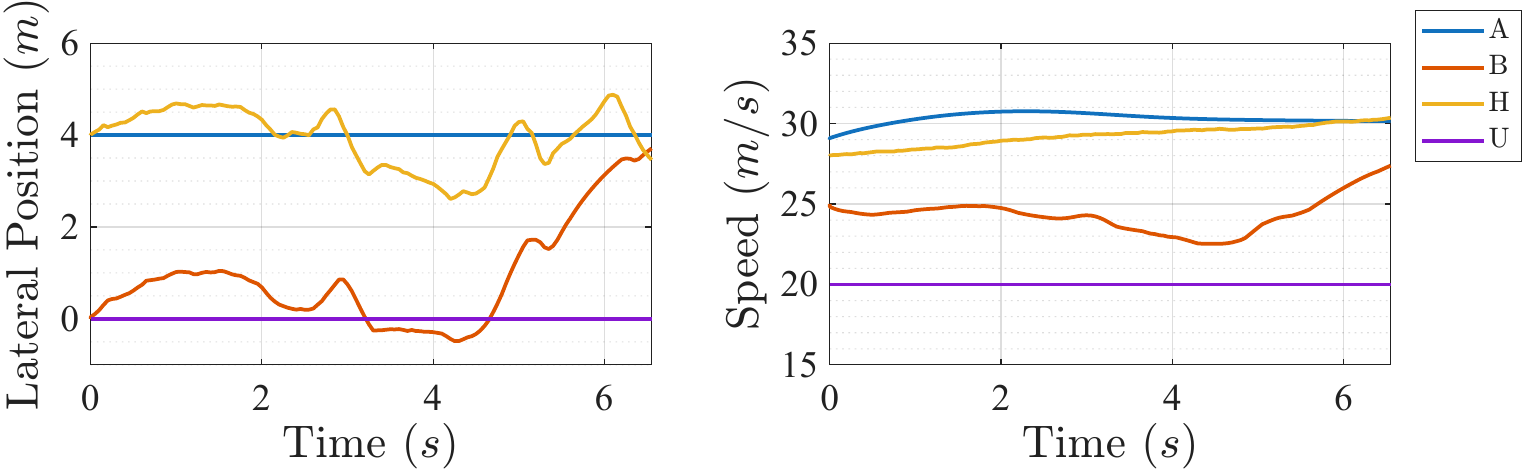}
\caption{Lateral position and speed under EU-FDI attacks using the proposed EDSR control framework.
}
\label{fig:4}
\end{figure}

\begin{figure}[H]
\centering
\includegraphics[width =0.46\textwidth]{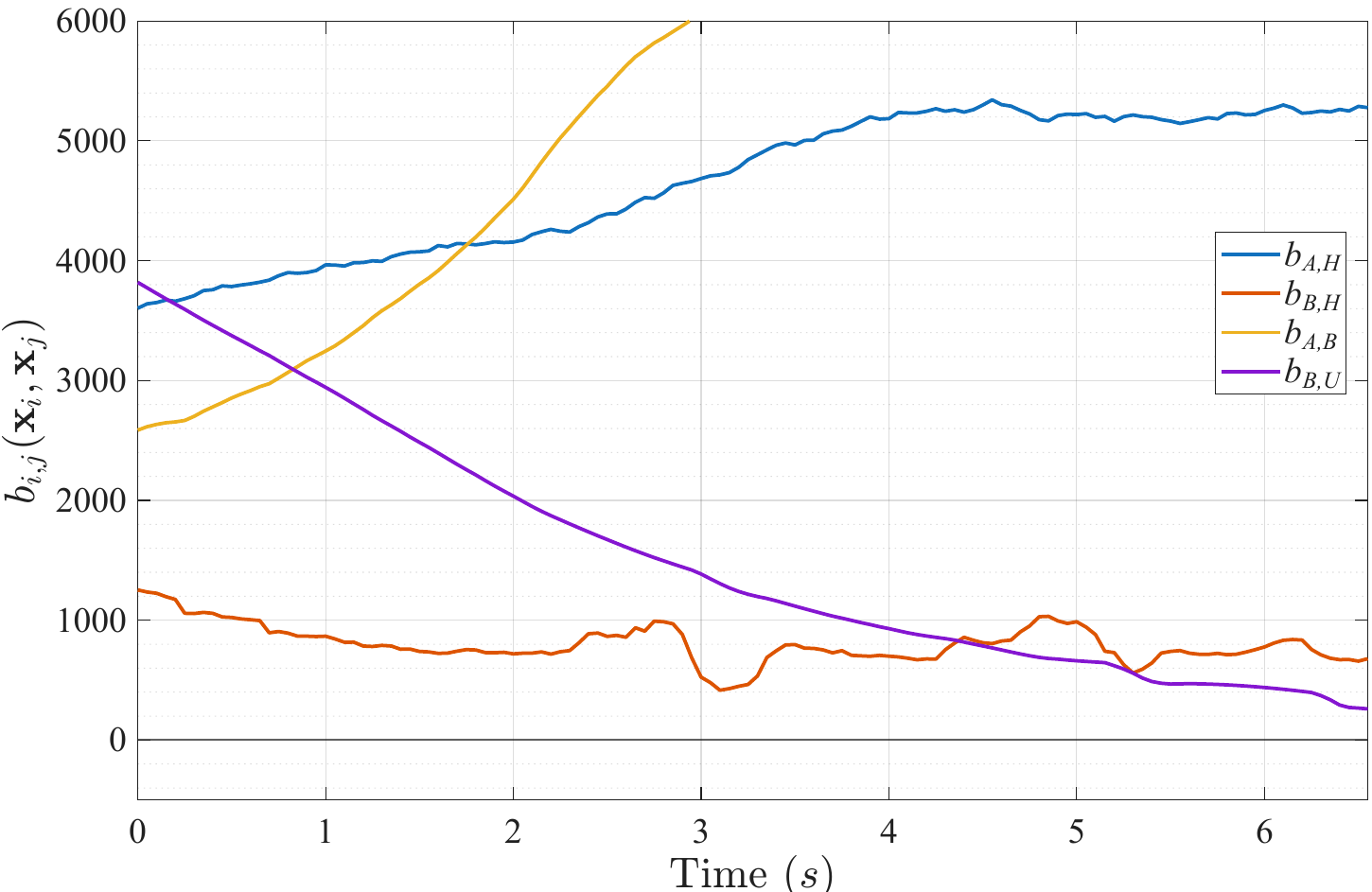}
\caption{Safety performance under EU-FDI attacks using the proposed EDSR control framework.
} 
\label{fig:5}
\end{figure}

\begin{figure}[H]
\centering
\includegraphics[width =0.46\textwidth]{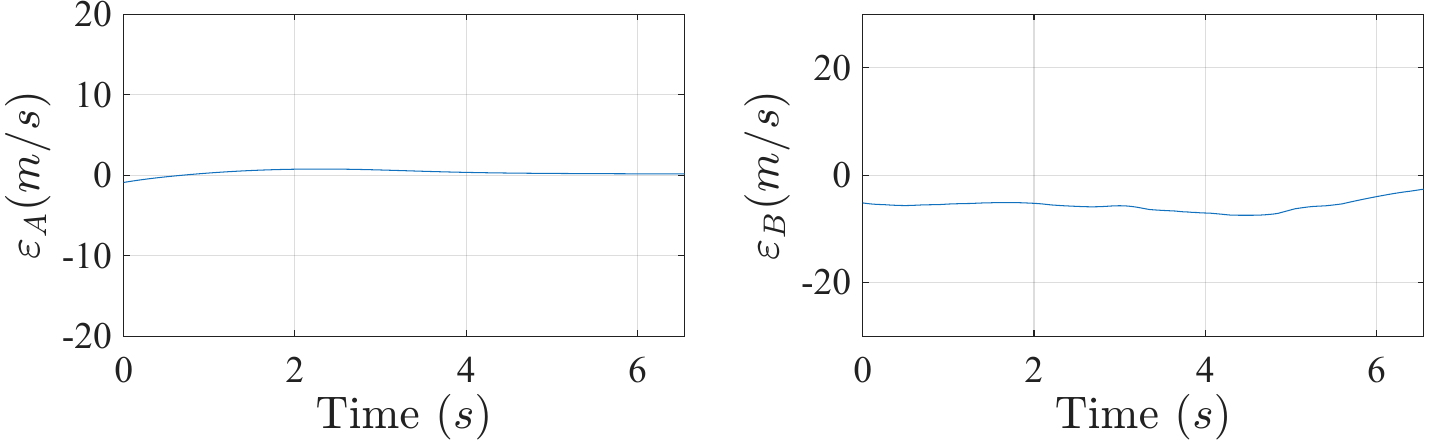}
\caption{Tracking error under EU-FDI attacks using the proposed EDSR control framework.
} 
\label{fig:6}
\end{figure}
\subsection{Comparative Study with Event-Driven CBF in \cite{xiao2022event}}

To evaluate the limitations of existing methods, we compare our EDSR framework with the event-driven approach from \cite{xiao2022event}, which lacks attack-resilient compensation. For fair comparison, the same EU-FDI attack formulation defined in \eqref{eq30a} and \eqref{eq30b} is applied. Figure~\ref{fig:7} shows the corresponding attack signals, and simulation results of the event-driven approach in \cite{xiao2022event} are presented in Figures~\ref{fig:9}--\ref{fig:11}.

\begin{figure}[H]
\centering
\includegraphics[width =0.46 \textwidth]{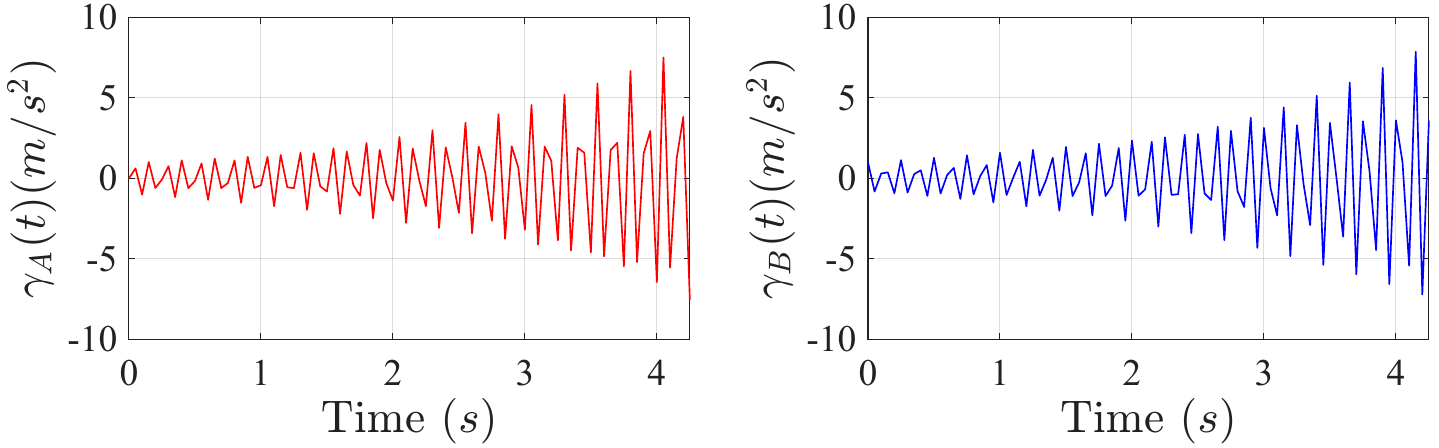}
\caption{EU-FDI Attack Signal.
}
\label{fig:7}
\end{figure}

\begin{figure}[H]
    \centering
    \begin{subfigure}[b]{0.2\textwidth}
        \includegraphics[width=\linewidth]{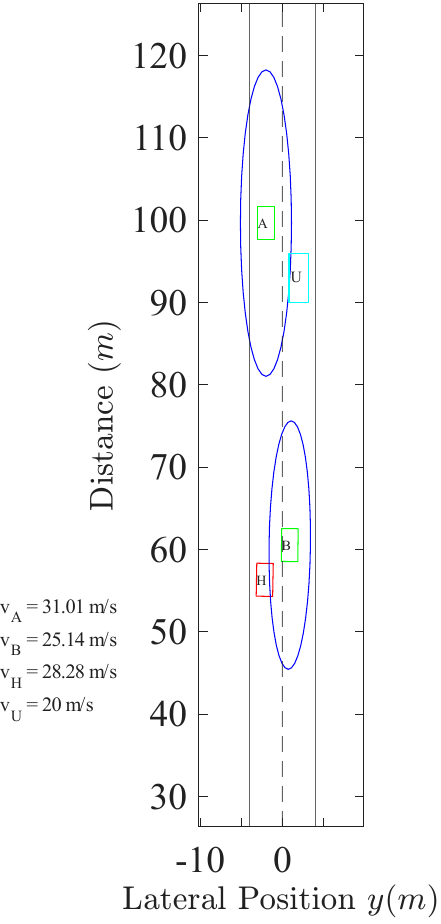}
        \caption{$t=1.50s$.}
        \label{fig:8a}
    \end{subfigure}
    \begin{subfigure}[b]{0.2\textwidth}
        \includegraphics[width=\linewidth]{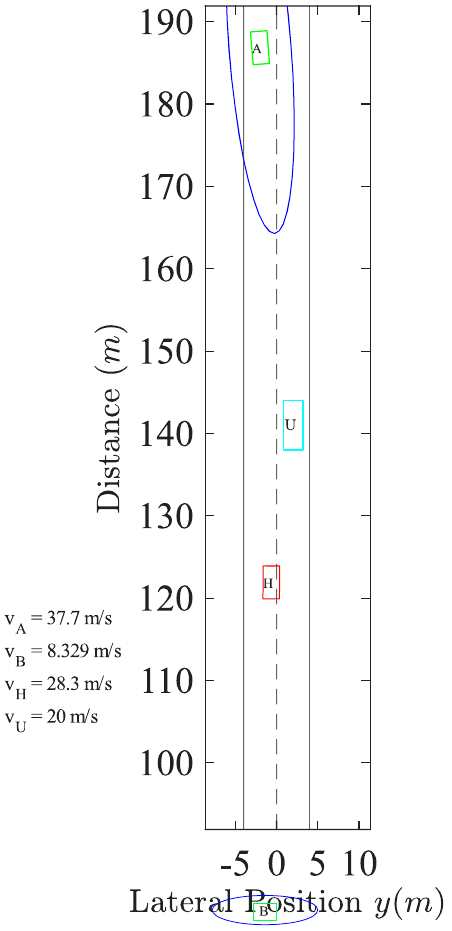}
        \caption{$t=3.69s$.}
        \label{fig:8b}
    \end{subfigure}
    \caption{Lane-changing maneuver under EU-FDI attacks using the event-driven approach in \cite{xiao2022event}.}
    \label{fig:8}
\end{figure}

Under this unified attack setting, Figure~\ref{fig:8} presents representative snapshots of the lane-changing process under the conventional event-driven method at $t=1.50$~s and $t=3.69$~s. 
The left configuration in Figure~\ref{fig:8a} ($t=1.50$~s) shows that the HDV has already intruded into the safety region of CAV $B$, indicating a violation of the safety constraint $b_{B,H}$, while CAV $B$ fails to initiate effective avoidance under EU-FDI attacks. 
The right configuration in Figure~\ref{fig:8b} ($t=3.69$~s) further reveals that CAV $B$ deviates significantly from the desired lane-changing trajectory, leading to degraded maneuver effectiveness and excessive fallback. 
These snapshots demonstrate progressive safety degradation in the absence of attack-resilient compensation. Figure~\ref{fig:9} shows the lateral positions and speeds of the vehicles. Without resilience to adversarial inputs, small HDV deviations induce drastic maneuvers and severe deceleration in CAV $B$, while CAV $A$ exhibits rapid and uncontrolled acceleration. The disturbance propagates through the traffic interaction, and CAV $B$ loses maneuver feasibility, resulting in premature termination at approximately $t=4.28$~s. This confirms that unmitigated EU-FDI attacks can rapidly drive the system toward unsafe behaviors. Figure~\ref{fig:10} shows the tracking errors $\varepsilon_i(t)$. Under escalating EU-FDI attacks, CAV $A$ shows sharp positive errors (up to about $+15~\mathrm{m/s}$), reflecting unwarranted acceleration, while CAV $B$ exhibits severe negative errors (minimum about $-30~\mathrm{m/s}$), indicating prolonged excessive deceleration. These large error magnitudes confirm that the conventional event-driven method fails to maintain stable speed tracking and produces highly unbalanced vehicle behaviors. Figure~\ref{fig:11} presents the CBF safety metrics $b_{i,j}$. The most critical degradation occurs in $b_{B,H}$, which drops below zero between approximately $t=2.15$~s and $t=2.60$~s and reaches about $-140$, clearly indicating loss of  forward invariance. Overall, the conventional event-driven CBFs method attempts to preserve safety through extreme acceleration or braking under attacks, which results in impractical trajectories and degraded traffic stability, thereby highlighting the necessity of the integrated resilience provided by the proposed EDSR framework.

\begin{figure}[h]
\centering
\includegraphics[width =0.46\textwidth]{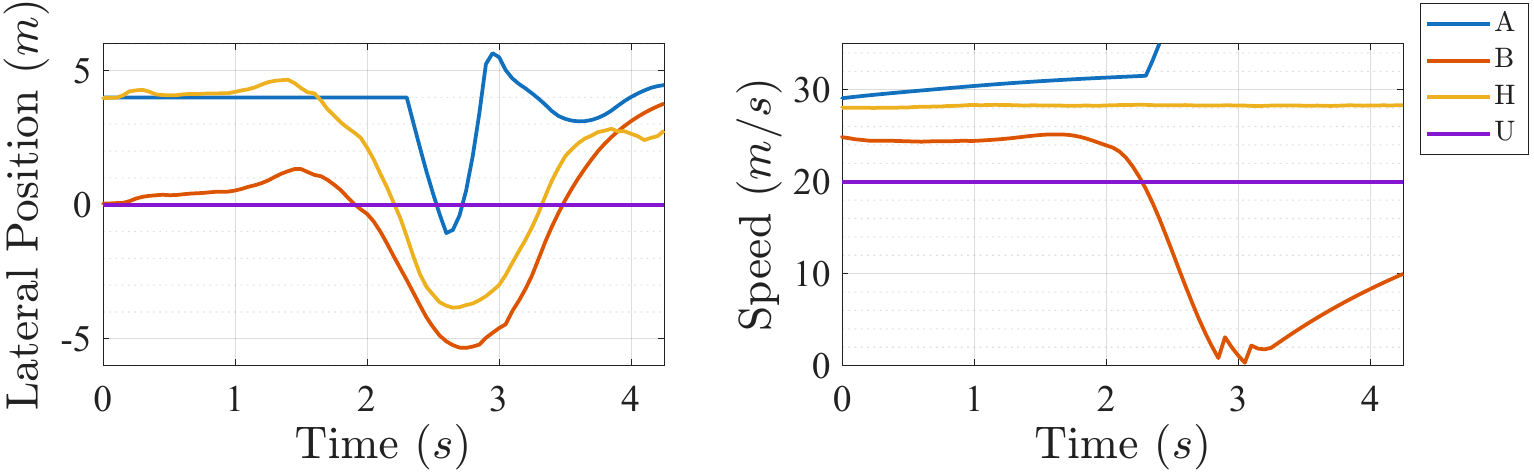}
\caption{Lateral position and speed under EU-FDI attacks using the event-driven approach in \cite{xiao2022event}.
}
\label{fig:9}
\end{figure}

\begin{figure}[h]
\centering
\includegraphics[width =0.46\textwidth]{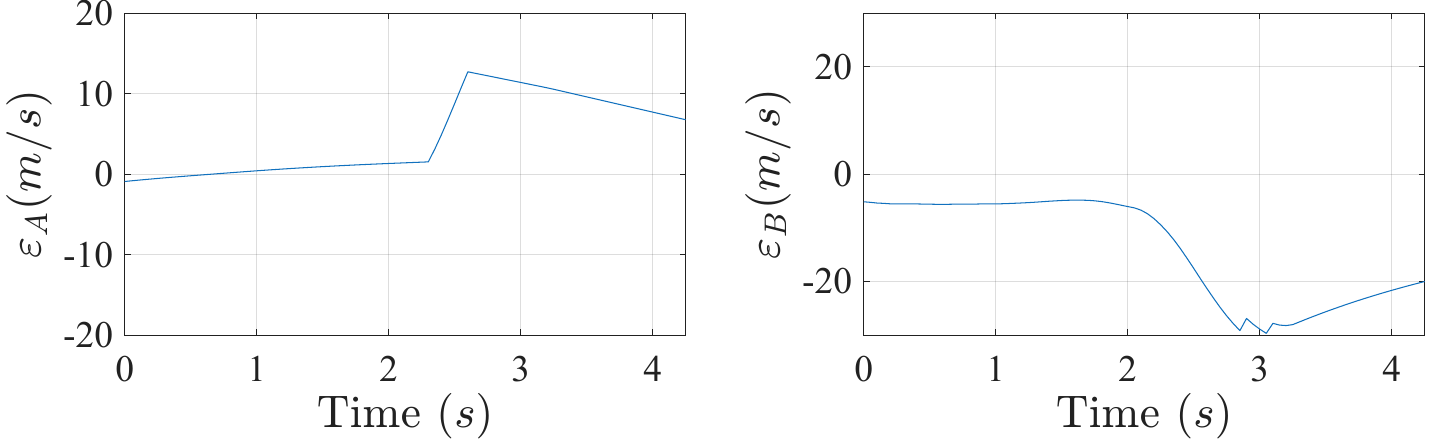}
\caption{Tracking error under EU-FDI attacks using the event-driven approach in \cite{xiao2022event}.
} 
\label{fig:10}
\end{figure}

\begin{figure}[h]
\centering
\includegraphics[width =0.46\textwidth]{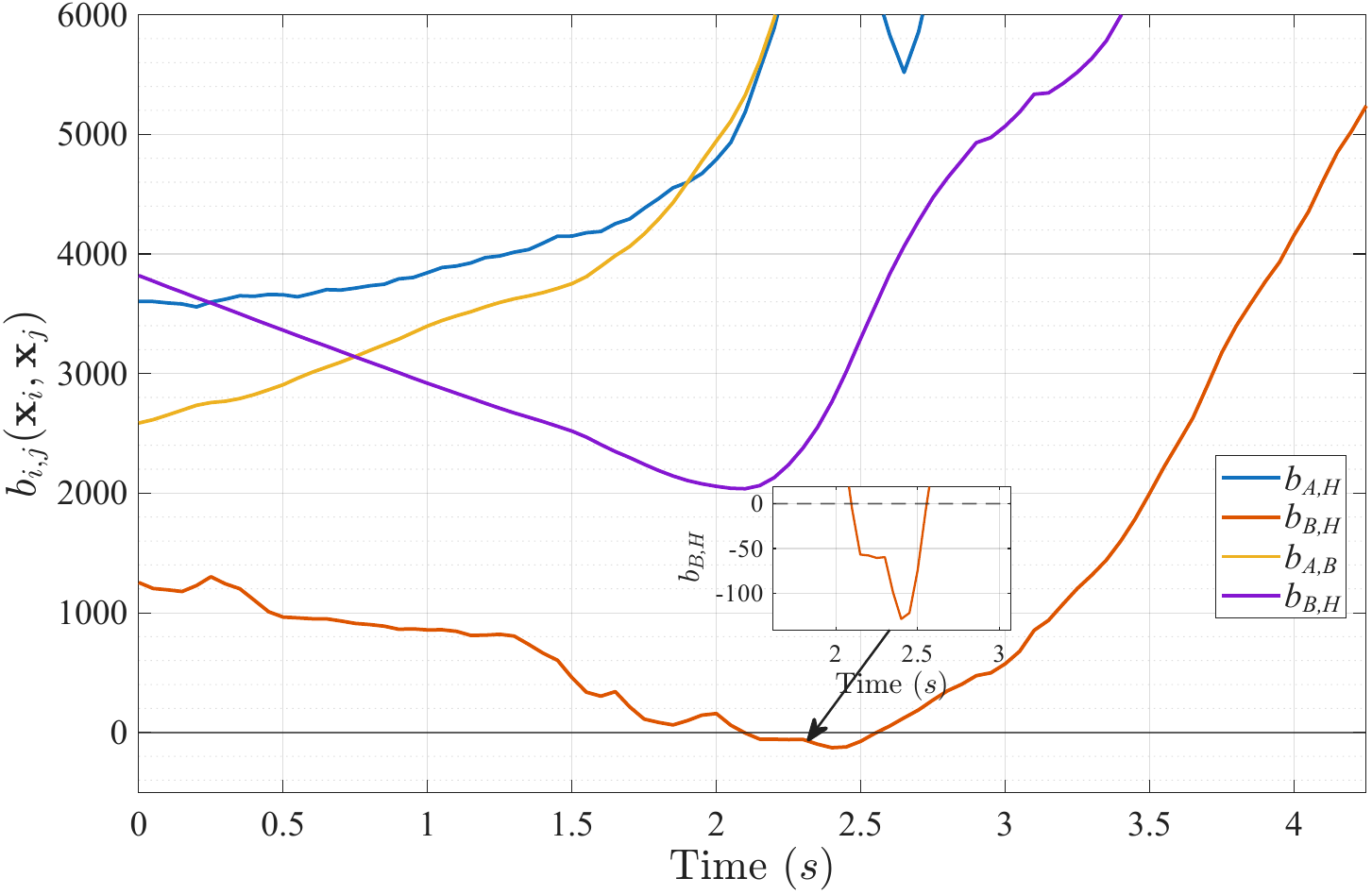}
\caption{Safety performance under EU-FDI attacks using the conventional event-driven approach in \cite{xiao2022event}.
} 
\label{fig:11}
\end{figure}

In this study, the simulation environment was designed to capture key dynamics of mixed traffic scenarios involving CAVs and HDVs under adversarial conditions. To improve realism, future work will incorporate more detailed vehicle models, including actuator dynamics, sensor noise, and communication delays. Additionally, we plan to integrate real-world traffic datasets and human driver behavior models to better reflect the variability and unpredictability of HDVs. These enhancements will further validate the robustness and applicability of the proposed control framework in practical deployment settings.

\section{Conclusion}
This paper formulates and addresses the SAR control problem for CAVs interacting with HDVs in mixed traffic environments. To jointly ensure safety and attack resilience under adversarial conditions, a novel EDSR control framework has been proposed that integrates event-driven CBFs and CLFs with data-driven estimation of HDV behaviors. This approach maintains provable collision avoidance guarantees for safe lane-changing maneuvers even in the presence of EU-FDI attacks, while also significantly reducing the computational burden through the event-driven design. Simulation results validate the effectiveness of the framework in ensuring real-time safety and attack-resilient control in uncertain and adversarial traffic scenarios.

\bibliographystyle{IEEEtran}
\bibliography{output}

\end{document}